\documentclass[a4paper,10pt]{article}
\usepackage{test,paralist,subcaption}
\usepackage[authoryear,round]{natbib}
\usepackage{hyperref}
\usepackage{doi}

\title{Susceptible-Infected-Recovered (SIR) Dynamics of COVID-19 and Economic Impact}
\author{Alexis Akira Toda\thanks{Department of Economics, University of California San Diego. Email: \href{mailto:atoda@ucsd.edu}{atoda@ucsd.edu}.}}

\numberwithin{equation}{section}
\numberwithin{thm}{section}

\newcommand{\Tincl}{14}
\newcommand{\betaMed}{0.29}
\newcommand{\ymaxMed}{28}

\begin{document}

\maketitle

\begin{abstract}
I estimate the Susceptible-Infected-Recovered (SIR) epidemic model for Coronavirus Disease 2019 (COVID-19). The transmission rate is heterogeneous across countries and far exceeds the recovery rate, which enables a fast spread. In the benchmark model, 28\% of the population may be simultaneously infected at the peak, potentially overwhelming the healthcare system. The peak reduces to 6.2\% under the optimal mitigation policy that controls the timing and intensity of social distancing. A stylized asset pricing model suggests that the stock price temporarily decreases by 50\% in the benchmark case but shows a W-shaped, moderate but longer bear market under the optimal policy.

\medskip


\medskip

\end{abstract}

\section{Introduction}

The novel coronavirus disease that was first reported in Wuhan, China in December 2019 (COVID-19) is quickly spreading around the world. As of {\today}, the total number of cases exceeds 460,000 and the disease has claimed more than 20,000 lives globally. Since March 2020, while new cases in China appears to have settled down, the number of cases are exponentially growing in the rest of the world. To prevent the spread of the new virus, many governments have introduced draconian measures such as restricting travel, ordering social distancing, and closing schools, bars, restaurants, and other businesses.

In a time of such extreme uncertainty, making economic decisions becomes challenging because pandemics are rare. The most recent comparable episode is the Spanish flu of 1918 \citep{Trilla2008}, so pandemics are likely to occur at most once during one's lifetime. Nevertheless, individuals need to make everyday decisions such as how to manage inventories of staples, how much to consume and save, when to buy or sell stocks, etc., and these decisions depend on the expectation of how long and severe the epidemic is. Governments must also make decisions such as to what extent imposing travel restrictions, social distancing, closure of schools and businesses, etc., and for how long \citep{Anderson_2020}.

When past experience or data are not so relevant in new situations such as the COVID-19 pandemic, simple mathematical models are useful in analyzing the current situation and predicting the near future. This paper aims to help decision making by building a mathematical epidemic model, estimating it using the up-to-date data of COVID-19 cases around the world, making out-of-sample predictions, and discussing optimal policy and economic impact. The model is the \cite{KermackMcKendrick1927} Susceptible-Infected-Recovered (SIR) model and is relatively simple. An infected individual interacts with other agents and transmits the disease at a certain rate if the other agent is susceptible. An infected individual also recovers (or dies) at a certain rate. The model can be described as a system of ordinary differential equations, which is nonlinear due to the interaction between the infected and susceptible. The behavior of the model is completely determined by the transmission rate ($\beta$), the recovery rate ($\gamma$), and the initial condition. Despite the nonlinearity, the model admits an exact analytical solution in parametric form \citep{HarkoLoboMak2014}, which is convenient for estimation and prediction. Using this model, I theoretically derive the condition under which an epidemic occurs and characterize the peak of the epidemic.

I next take this model to the data. Because the situation and policies surrounding COVID-19 is rapidly evolving, I use the most recent two weeks ({\Tincl} days) of cases and estimate the model parameters by nonlinear least squares. Except for China, Japan, and Korea, which are early epicenters of the outbreak, the transmission rate $\beta$ is around 0.2--0.4 and heterogeneous across countries. The estimated transmission rates far exceed the recovery rate $\gamma$, which is about 0.1 based on the clinical course of COVID-19. Due to the high transmission rate and lack of herd immunity, in the absence of mitigation measures such as social distancing, the virus spreads quickly and may infect around 30 percent of the population at the peak of the epidemic. Using the model, I conduct an experiment where the government introduces temporary mitigation measures and succeeds in reducing the transmission rate. If the mitigation measures are taken too early, the peak is delayed but the epidemic restarts with no effect on the peak because the population does not acquire herd immunity. Assuming the government can take drastic measures up to 12 weeks, the optimal policy is start mitigation measures once the number of cases reaches 6.3\% of the population. Under the optimal policy, the peak infection rate reduces to 6.2\%. Therefore unless vaccines are expected to be developed in the near future, the draconian measures currently taken in many countries may be suboptimal, and it may be desirable to postpone them.

To evaluate the potential economic impact of COVID-19, I build a stylized production-based asset pricing model. Capitalists hire labor at competitive markets and infected workers are unable to work. Because the epidemic (temporarily) drastically reduces the labor supply, output goes down and the model calibration suggests that the stock market crashes by 50\% during the epidemic, though the crash is short-lived. Under the optimal policy, the stock price exhibits a W-shaped pattern and remains about 10\% undervalued than the steady state for about half a year.

\section{SIR epidemic model}\label{sec:SIR}

I first present the compartment model of epidemics following \cite{KermackMcKendrick1927}.

The society consists of $N$ individuals, among which $S$ are susceptible to an infectious disease (they are neither infected nor have immunity) and $I$ are infected. (We ignore population growth because an epidemic occurs in a relatively short interval.) Let $R=N-S-I$ be the number of individuals who are immune (possibly because they are vaccinated, infected and recovered, or dead). Suppose that individuals meet each other randomly, and conditional of an infected individual meeting a susceptible individual, the disease is transmitted with some probability. Let $\beta>0$ be the rate at which an infected individual meets a person and transmits the disease if susceptible. Let $\gamma>0$ be the rate at which an infected individual recovers or dies. Then the following differential equations hold.
\begin{subequations}\label{eq:SIR}
\begin{align}
\diff S/\diff t&=-\beta SI/N, \label{eq:SIR.s}\\
\diff I/\diff t&=\beta SI/N-\gamma I, \label{eq:SIR.i}\\
\diff R/\diff t&=\gamma I. \label{eq:SIR.r}
\end{align}
\end{subequations}
To see why \eqref{eq:SIR.s} holds, note that an infected individual can transmit to $\beta$ people per unit of time if all of them are susceptible, but the probability of meeting a susceptible individual is only $S/N$. Thus, $I$ infected individuals can transmit to $I\times \beta\times (S/N)=\beta SI/N$ individuals per unit of time. \eqref{eq:SIR.i} holds because the change in the number of infected individuals equals the newly infected minus closed cases (either due to recovery or death).

Letting $x=S/N$, $y=I/N$, $z=R/N$ be the fraction of susceptible, infected, and recovered individuals in the society, dividing all equations in \eqref{eq:SIR} by $N$, we obtain
\begin{subequations}\label{eq:xyz}
\begin{align}
\dot{x}&=-\beta xy, \label{eq:xyz.x}\\
\dot{y}&=\beta xy-\gamma y, \label{eq:xyz.y}\\
\dot{z}&=\gamma y,\label{eq:xyz.z}
\end{align}
\end{subequations}
where $\dot{x}=\diff x/\diff t$. Although the system of differential equations \eqref{eq:xyz} is nonlinear, \cite{HarkoLoboMak2014} obtain an exact analytical solution in parametric form.

\begin{prop}\label{prop:HLM}
Let $x(0)=x_0>0$, $y(0)=y_0>0$, $z(0)=z_0\ge 0$ be given, where $x_0+y_0+z_0=1$. Then the solution to \eqref{eq:xyz} is parametrized as
\begin{subequations}\label{eq:HLM}
\begin{align}
x(t)&=x_0v, \label{eq:HLM.x}\\
y(t)&=\frac{\gamma}{\beta}\log v-x_0v+x_0+y_0, \label{eq:HLM.y}\\
z(t)&=-\frac{\gamma}{\beta}\log v+z_0, \label{eq:HLM.z}
\end{align}
\end{subequations}
where
\begin{equation}
t=\int_v^1\frac{\diff \xi}{\xi(\beta x_0(1-\xi)+\beta y_0+\gamma\log \xi)}. \label{eq:v}
\end{equation}
\end{prop}

\begin{proof}
See Equations (26)--(29) in \cite{HarkoLoboMak2014}. The parametrization has been changed slightly for convenience.
\end{proof}

Using Proposition \ref{prop:HLM}, we can study the qualitative properties of the epidemic.

\begin{prop}\label{prop:epidemic}
Let everything be as in Proposition \ref{prop:HLM}. Then the followings are true.
\begin{enumerate}
\item In the long run, fraction $v^*\in (0,1)$ of susceptible individuals will not be infected (fraction $1-v^*$ infected), where $v^*$ is the unique solution to
\begin{equation}
x_0(1-v)+y_0+\frac{\gamma}{\beta}\log v=0.\label{eq:vstar}
\end{equation}
\item If $\beta x_0\le \gamma$, then $\diff y/\diff t\le 0$: there is no epidemic. Furthermore, $v^*\to 1$ as $y_0\to 0$.
\item If $\beta x_0>\gamma$, then there is an epidemic. The number of infected individuals reaches the maximum when $\beta x(t_{\max})=\gamma$, at which point the fraction
\begin{equation}
y_{\max}=y(t_{\max})=\frac{\gamma}{\beta}\log \frac{\gamma}{\beta x_0}-\frac{\gamma}{\beta}+x_0+y_0\label{eq:ymax}
\end{equation}
of population is infected. The maximum infection rate $y_{\max}$ is increasing in $x_0,y_0$ and decreasing in $\gamma/\beta$.
\end{enumerate}
\end{prop}

\begin{proof}
Let $f(v)=x_0(1-v)+y_0+\frac{\gamma}{\beta}\log v$ for $v\in (0,1]$. Then \eqref{eq:v} implies
\begin{equation}
t=\int_{v(t)}^1 \frac{\diff \xi}{\beta\xi f(\xi)}.\label{eq:vt}
\end{equation}
Since $f(1)=y_0>0$, it must be $v(0)=1$. The definite integral \eqref{eq:vt} is well-defined in the range $f(v)>0$. Since
\begin{align*}
f'(v)&=-x_0+\frac{\gamma}{\beta v},\\
f''(v)&=-\frac{\gamma}{\beta v^2}<0,
\end{align*}
$f$ is concave so the set $V=\set{v\in (0,1]|f(v)>0}$ is an interval. Since $f(v)\to -\infty$ as $v\downarrow 0$, we have $V=(v^*,1]$ for $v^*\in (0,1)$, where $v^*$ solves \eqref{eq:vstar}. Because $f$ can be approximated by a linear function around $v^*$, we get
$$\infty=\int_{v^*}^1\frac{\diff \xi}{\beta \xi f(\xi)},$$
so $v(\infty)=v^*$. Using \eqref{eq:HLM.x}, in the long run fraction $x(\infty)/x_0=v^*$ of susceptible individuals are not infected.

Since $f(v)>0$ on $V=(v^*,1]$, we have $v(t)\in (v^*,1]$ for all $t\ge 0$. By \eqref{eq:vt}, $v(t)$ is clearly decreasing in $t$. If $\beta x_0\le \gamma$, it follows from \eqref{eq:HLM.y} that
$$\dot{y}=\left(\frac{\gamma}{\beta v}-x_0\right)\dot{v}=\frac{\gamma-\beta x_0 v}{\beta v}\dot{v}\le 0$$
because $\dot{v}\le 0$ and $v\le 1$ implies $\gamma-\beta x_0v\ge \gamma-\beta x_0\ge 0$. Since $f(1)=0$ when $y_0=0$, $f'(1)=-x_0+\gamma/\beta\ge 0$ if $\beta x_0\le \gamma$, and $f''(v)<0$, it must be $v^*\to 1$ as $y_0\to 0$.

Finally, assume $\beta x_0>\gamma$. Then $\dot{y}(0)=(\beta x_0-\gamma)y_0>0$, so $y(t)$ initially increases. By \eqref{eq:xyz.y}, $y(t)$ reaches the maximum when $0=\dot{y}=\beta xy-\gamma y\iff x=\gamma/\beta$. Using \eqref{eq:HLM.x}, this is achieved when $\gamma/\beta=x_0v\iff v=\frac{\gamma}{\beta x_0}$. Substituting into \eqref{eq:HLM.y}, we obtain \eqref{eq:ymax}. Letting
$$y(\theta,x_0,y_0)=\theta \log\frac{\theta}{x_0}-\theta+x_0+y_0$$
for $\theta=\gamma/\beta$, it follows from simple algebra that
\begin{align*}
\partial y/\partial y_0&=1,\\
\partial y/\partial x_0&=-\frac{\theta}{x_0}+1=\frac{\beta x_0-\gamma}{\beta x_0}>0,\\
\partial y/\partial \theta&=\log \frac{\gamma}{\beta x_0}<0,
\end{align*}
so $y_{\max}$ is increasing in $x_0,y_0$ and decreasing in $\theta=\gamma/\beta$.
\end{proof}

Proposition \ref{prop:epidemic} has several policy implications for dealing with epidemics. First, the policy maker may want to prevent an epidemic. This is achieved when the condition $\beta x_0\le \gamma$ holds. Since before the epidemic the fraction of infected individuals $y_0$ is negligible, we can rewrite the no-epidemic condition as $\beta(1-z_0)\le \gamma$. Unlike bacterial infections, for which a large variety of antibiotics are available, there is generally no curative care for viral infections.\footnote{Currently, the only viruses against which antiviral drugs are available are the human immunodeficiency virus (HIV), herpes, hepatitis, and influenza viruses. See \cite{Razonable2011} for a review of treatments of the latter three viruses.} Therefore the recovery/death rate $\gamma$ is generally out of control. Hence the only way to satisfy the no-epidemic condition $\beta(1-z_0)\le \gamma$ is either
\begin{inparaenum}[(i)]
\item control transmission (reduce $\beta$), for example by washing hands, wearing protective gear, restricting travel, or social distancing, or
\item immunization (increase $z_0$).
\end{inparaenum}
The required minimum immunization rate to prevent an epidemic is $z_0=1-\gamma/\beta$.

Second, the policy maker may want to limit the economic impact once an epidemic occurs. Because the supply of healthcare services is inelastic in the short run, it is important to keep the maximum infection rate $y_{\max}$ in \eqref{eq:ymax} within the capacity of the existing healthcare system. This is achieved by lowering the transmission rate $\beta$.

\section{Estimation and prediction}\label{sec:estim}

In this section I estimate the SIR model in Section \ref{sec:SIR} and use it to predict the evolution of the COVID-19 pandemic.

\subsection{Data}

The number of cases of COVID-19 is provided by Center for Systems Science and Engineering at Johns Hopkins University (henceforth CSSE). The cumulative number of confirmed cases and deaths can be downloaded from the GitHub repository.\footnote{\url{https://github.com/CSSEGISandData/COVID-19/tree/master/csse_covid_19_data/csse_covid_19_time_series}} The time series starts on January 22, 2020 and is updated daily. Because countries are added as new cases are reported, the cross-sectional size increases every day. For the majority of countries, the CSSE data are at country level. However, for some countries such as Australia, Canada, and China, regional data at the level of province or state are available. In such countries, I aggregate across regions and use the country level data. Figure \ref{fig:data_Early} shows the number of COVID-19 cases in early epicenters, namely China, Iran, Italy, Japan, and Korea.

\begin{figure}[!htb]
\centering
\includegraphics[width=0.7\linewidth]{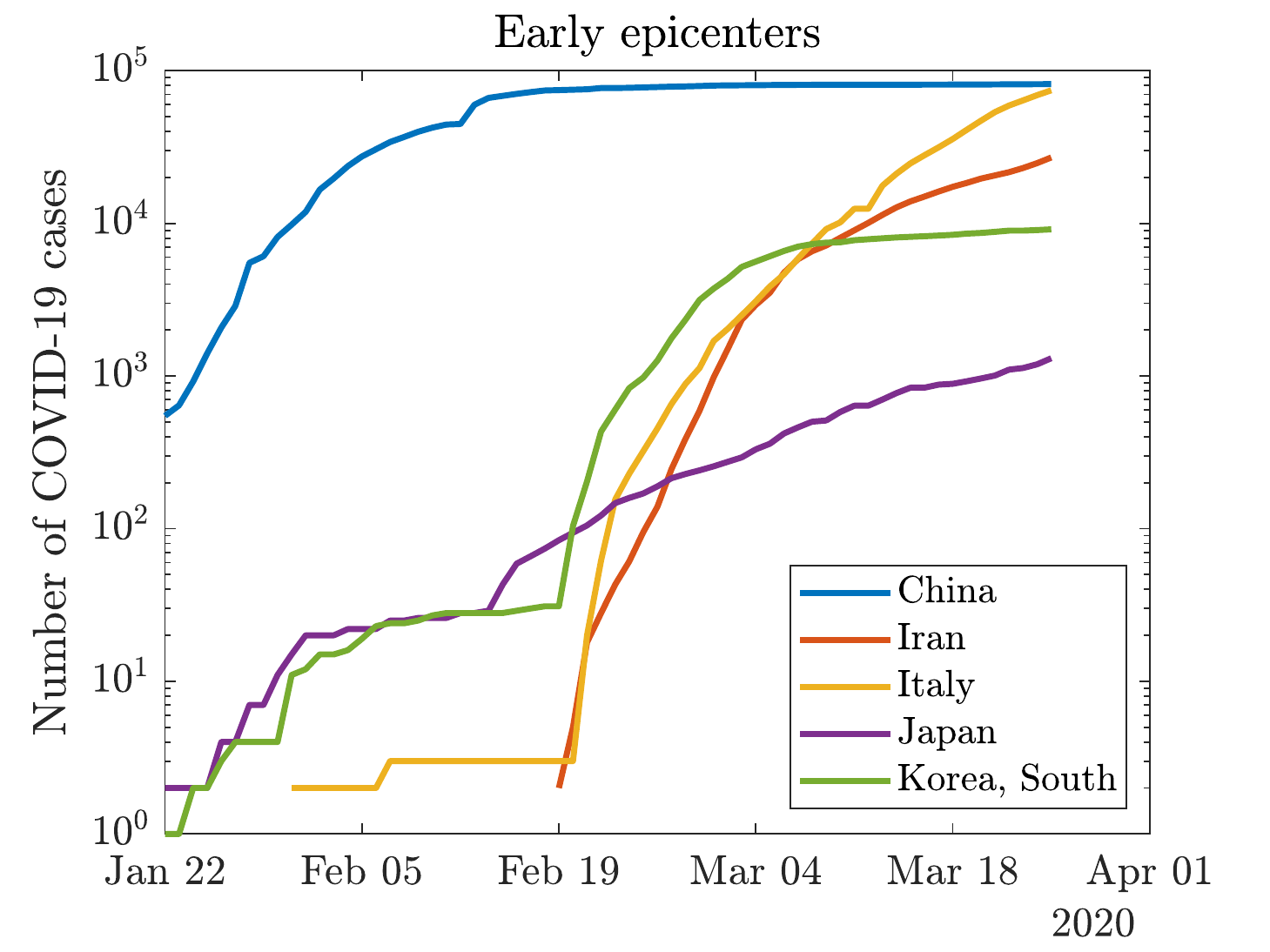}
\caption{Number of COVID-19 cases in early epicenters.}\label{fig:data_Early}
\end{figure}

\subsection{Estimation}

Estimation of the model poses significant challenges because the situation of COVID-19 is rapidly evolving. The model parameters are likely time-varying because new policies are introduced on a day-to-day basis, temperature and weather may affect the virus activity, and the virus itself my genetically mutate. For this reason, I only use the data from the two most recent weeks ({\Tincl} days).

I estimate the model parameters by nonlinear least squares, minimizing the distance between model outputs $(x,y,z)$ and data. Because the CSSE data only contains confirmed cases and deaths, but the SIR model abstracts from death, I define $c=y+z=1-x$ to be the fraction of infected or recovered cases in the model. The counterpart in the data is $\widehat{c}=C/N$, where $C$ is the number of confirmed cases and $N$ is population.\footnote{I use the 2015 population data from World Bank at \url{https://data.world/worldbank/total-population-per-country}.} Because the number of cases grows by many orders of magnitude within a short period of time, I define the loss function using log cases:
\begin{equation}
L(\beta,\gamma,y_0,z_0)=\sum_t\left(\log \widehat{c}(t)-\log c(t)\right)^2.\label{eq:lossfunc}
\end{equation}

Since I only include $c$ in the loss function \eqref{eq:lossfunc}, the parameters $\gamma$ and $z_0$, which govern the dynamics of fraction of recovered $z$, are not identified. Therefore I exogenously fix these two parameters. For the recovery rate $\gamma$, because the majority of patients with COVID-19 experience mild symptoms that resemble a common cold or influenza \citep{Zhou_2020}, which takes about 10 days to recover, I set $\gamma=1/10=0.1$. For $z_0$, I set it to one divided by population.\footnote{This number is likely a significant underestimate, but the results are not sensitive to $z_0$ as long as it is small.} Although the fraction of cases $c(t)$ is likely significantly underestimated because infected individuals do not appear in the data unless they are tested, it does not cause problems for estimating the parameter of interest (the transmission rate $\beta$) because under-reporting is absorbed by the constant $y_0$ in \eqref{eq:HLM.y}, which only affects the onset of the epidemic by a few weeks without changing the overall dynamics (see Figure \ref{fig:SIR_example}). To sum up, I estimate the remaining parameters $\beta$ and $y_0$ by numerically minimizing the loss function \eqref{eq:lossfunc}. Standard errors are calculated using the asymptotic theory of $M$-estimators. See Appendix \ref{sec:solve} for the solution algorithm of the SIR model.

\subsection{Results}

I estimate the SIR model for all countries that meet the following inclusion criteria:
\begin{inparaenum}[(i)]
\item the number of confirmed cases as of {\today} exceeds 1,000, and
\item the number of confirmed cases at the beginning of the estimation sample exceeds 10.
\end{inparaenum}
These countries are mostly early epicenters (China, Japan, Korea), European countries, and North America. Table \ref{t:SIR_estim} shows the estimated transmission rate ($\beta$), its standard error, the fraction of infected individuals at the peak ($y_{\max}$), number of days to reach the peak ($t_{\max}$), and the fraction of the population that is eventually infected. Figure \ref{fig:Italy} shows the time evolution of COVID-19 cases in Italy, which is the earliest epicenter outside East Asia.

\begin{table}[!htb]
\centering
\caption{Estimation of SIR model.}\label{t:SIR_estim}
\begin{tabular}{lrrrrr}
%\toprule
\hline
Country & $\beta$ & s.e. & $y_{\max}$ (\%) & $t_{\max}$ (days) & Total (\%) \\ 
\hline 
Australia & 0.29 & 0.052 & 29 & 67 & 93 \\ 
Austria & 0.29 & 0.005 & 29 & 57 & 93 \\ 
Belgium & 0.27 & 0.112 & 26 & 64 & 91 \\ 
Brazil & 0.37 & 0.002 & 37 & 60 & 97 \\ 
Canada & 0.33 & 0 & 33 & 60 & 96 \\ 
Chile & 0.37 & 0.223 & 37 & 54 & 97 \\ 
China & 0.0012 & 0 & 0.0059 & 0 & 0.006 \\ 
Czechia & 0.29 & 0.003 & 29 & 64 & 93 \\ 
Denmark & 0.12 & 0.001 & 1.5 & 315 & 31 \\ 
Ecuador & 0.48 & 0 & 46 & 42 & 99 \\ 
France & 0.24 & 0.005 & 22 & 74 & 88 \\ 
Germany & 0.28 & 0.005 & 28 & 60 & 93 \\ 
Iran & 0.11 & 0.002 & 0.49 & 470 & 19 \\ 
Ireland & 0.35 & 0.009 & 35 & 50 & 96 \\ 
Israel & 0.3 & 0.101 & 30 & 62 & 94 \\ 
Italy & 0.19 & 0.002 & 13 & 91 & 76 \\ 
Japan & 0.077 & 0.003 & 0.00051 & 0 & 0.0022 \\ 
Korea, South & 0.02 & 0 & 0.015 & 0 & 0.019 \\ 
Luxembourg & 0.42 & 0.011 & 42 & 36 & 98 \\ 
Malaysia & 0.26 & 0.01 & 24 & 80 & 90 \\ 
Netherlands & 0.25 & 0.002 & 24 & 69 & 90 \\ 
Norway & 0.15 & 0.001 & 7 & 144 & 60 \\ 
Pakistan & 0.31 & 0.006 & 31 & 76 & 94 \\ 
Poland & 0.31 & 0.002 & 31 & 69 & 94 \\ 
Portugal & 0.37 & 0.004 & 37 & 48 & 97 \\ 
Spain & 0.28 & 0.118 & 27 & 57 & 92 \\ 
Sweden & 0.15 & 0.002 & 6 & 173 & 57 \\ 
Switzerland & 0.28 & 0.169 & 27 & 55 & 92 \\ 
US & 0.38 & 0.001 & 39 & 48 & 98 \\ 
United Kingdom & 0.29 & 0.088 & 29 & 64 & 94 \\ 
%\bottomrule 
\hline
\end{tabular}
\caption*{\footnotesize Note: The table presents the estimation results of the SIR model in Section \ref{sec:SIR}. $\beta$ (s.e.): the transmission rate and standard error; $y_{\max}$: the fraction of infected individuals at the peak in \eqref{eq:ymax}; $t_{\max}$: the number of days to reach the peak; ``Total'': the fraction of the population that is eventually infected.}
\end{table}

\begin{figure}[!htb]
\centering
\includegraphics[width=0.7\linewidth]{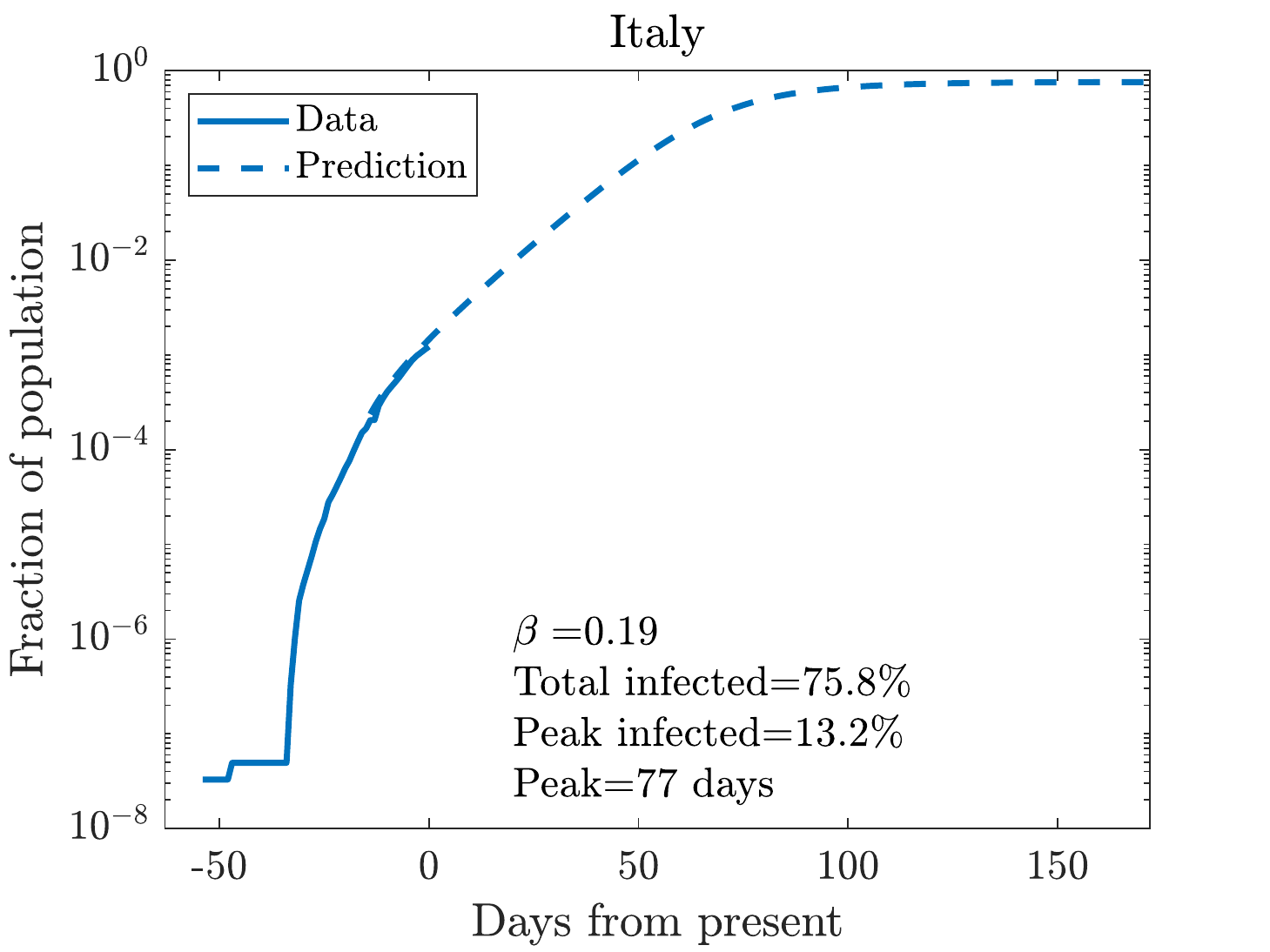}
\caption{Time evolution of COVID-19 cases in Italy.}\label{fig:Italy}
\end{figure}

We can make a few observations from Table \ref{t:SIR_estim}. First, the estimated transmission rates are heterogeneous across countries. While $\beta$ is low in China, the origin of COVID-19, and the neighboring countries (Japan and Korea), where the virus spread first, $\beta$ is very high at around 0.2--0.4 in other countries and the no-epidemic condition $\beta x_0\le \gamma$ fails. Despite the short time series ({\Tincl} days), the transmission rate is precisely estimated in most countries. Although current data is insufficient to draw any conclusion, there are a few possible explanations for the heterogeneity of $\beta$. First, the transmission rate $\beta$ may artificially appear high in later epicenters such as Europe and North America just because these countries were slow in adopting tests of COVID-19 and the testing (hence reporting) rate is increasing. Second, the heterogeneity in $\beta$ may be due to the fact that early epicenters have already taken mitigation measures of COVID-19. For example, while Japan closed all schools starting on March 2, many states in US have implemented similar measures such as closing schools, bars, and restaurants only around March 16, so we may not have yet seen the effect of such policies. Finally, it is possible that there are cultural differences. For example, school children in Japan are taught to wash their hands before eating and to gargle after returning home, which they practice, and (from personal experience) Japanese cities tend to be much cleaner than most cities in the world.

Second, according to the model, countries other than China, Japan, and Korea are significantly affected by the epidemic. If the current trend in the transmission rate $\beta$ continues, the epidemic will peak in May 2020, at which point around 30 percent of the population will be infected by the virus simultaneously. By the time the epidemic ends, more than 90 percent of the population is eventually infected. These numbers can be used to do a back-of-the-envelope calculation of health outcomes. In February 2020, the cruise ship Diamond Princess was put under quarantine for two weeks after COVID-19 was detected. All passengers were tested and tracked, among whom 712 tested positive and 8 died. Although this is not a representative sample because the cruise ship passengers tend to be older and wealthier, the mortality of COVID-19 should be around 1\% for this group and possibly lower for the general population. \cite{Zhou_2020} document that 54 patients died among 191 that required hospitalization in two hospitals in Wuhan. Therefore the ratio of patients requiring hospitalization to death is $191/54=3.56$. Thus, based on the model, the fraction of people requiring hospitalization at the peak is $y_{\max}\times 0.01 \times 3.56=1.0\%$ assuming $y_{\max}=\ymaxMed\%$, the median value in Table \ref{t:SIR_estim}.

\subsection{Optimal mitigation policy}

Using the estimated model parameters, we can predict the course of the epidemic. For this exercise, I consider the following scenario. The epidemic starts with the initial condition $(y_0,z_0)=(10^{-8},0)$. The benchmark transmission rate is set to the median value in Table \ref{t:SIR_estim}, which is $\beta=\betaMed$. When the number of total cases $c=y+z$ exceeds $10^{-5}$, the government introduces mitigation measures such as social distancing, and the transmission rate changes to either $\beta=0.2$ or $\beta=0.1$.\footnote{Using high-frequency data on influenza prevalence and quasi-experimental variation in mitigation measures, \cite{Adda2016} documents that school closures and travel restrictions are generally not cost-effective.} Mitigation measures are lifted after 12 weeks and the transmission rate returns to the benchmark value. I also consider the optimal mitigation policy, where the government chooses the threshold of cases $\bar{c}$ to introduce mitigation measures as well as the transmission rate $\beta$ to minimize the maximum infection rate $y_{\max}$.

Figure \ref{fig:mitigation} shows the fraction of infected and recovered over time. When the government introduces early but temporary mitigation measures (left panel), the epidemic is delayed but the peak is unaffected. This is because the maximum infection rate $y_{\max}$ in \eqref{eq:ymax} is mostly determined by $\beta$ and $\gamma$ since $(x_0,y_0)\approx (1,0)$, and the epidemic persists until the population acquires herd immunity so that the no-epidemic condition $\beta x\le \gamma$ holds. While early drastic mitigation measures might be useful to buy time to develop a vaccine, they may not be effective in mitigating the peak unless they are permanent.

The right panel in Figure \ref{fig:mitigation} shows the course of the epidemic under the optimal policy, which is to introduce mitigation measures such that $\beta=0.13$ when the number of cases reaches $\bar{c}=6.3\%$ of the population. Under this scenario, only $y_{\max}=6.2\%$ of the population is simultaneously infected at the peak as opposed to 28\% under the benchmark scenario. The intuition is that by waiting to introduce mitigation measures, a sufficient fraction of the population is infected (and acquires herd immunity) and thus reduces the peak.

\begin{figure}[!htb]
\centering
\begin{subfigure}{0.48\linewidth}
\includegraphics[width=\linewidth]{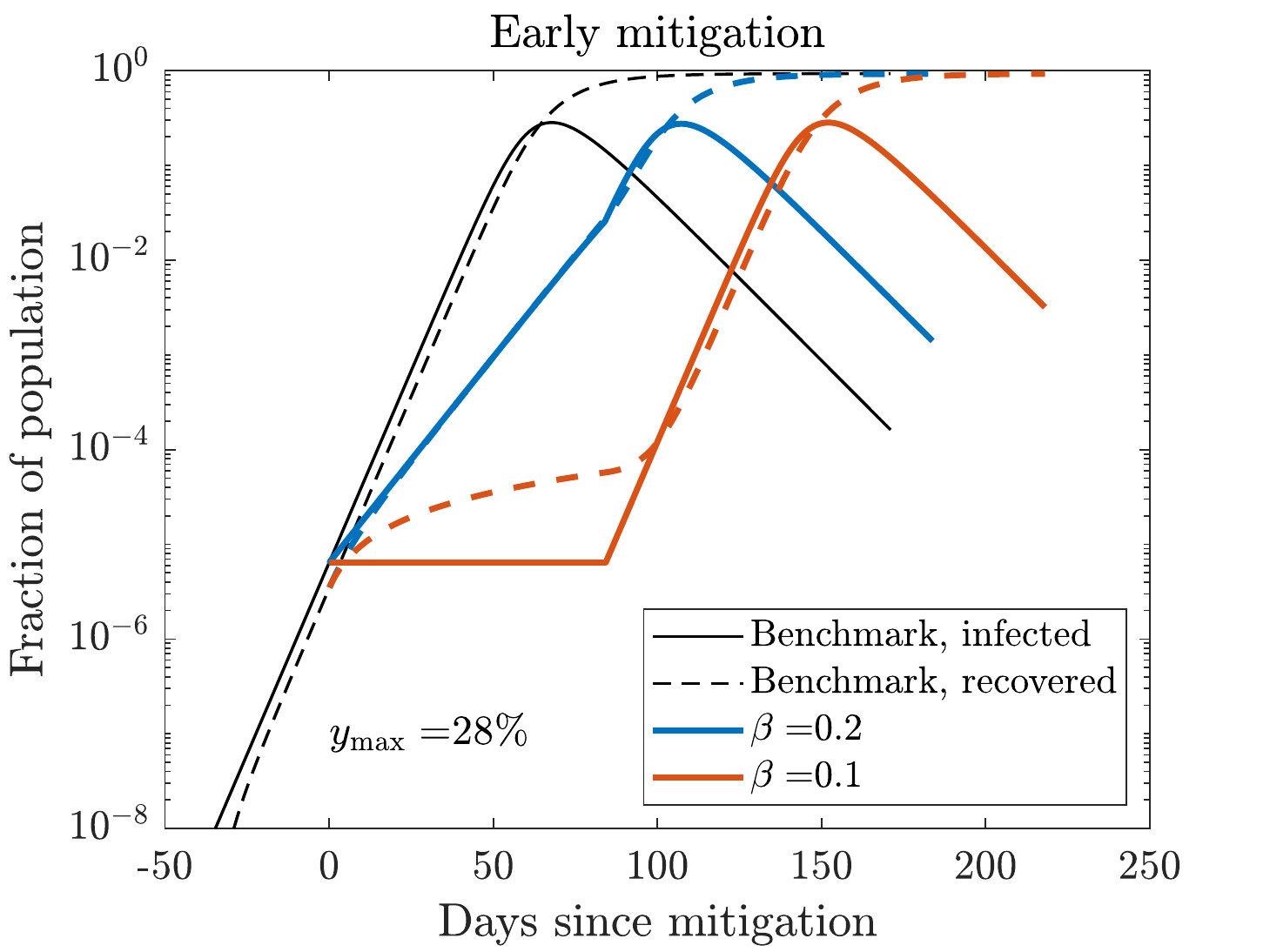}
\end{subfigure}
\begin{subfigure}{0.48\linewidth}
\includegraphics[width=\linewidth]{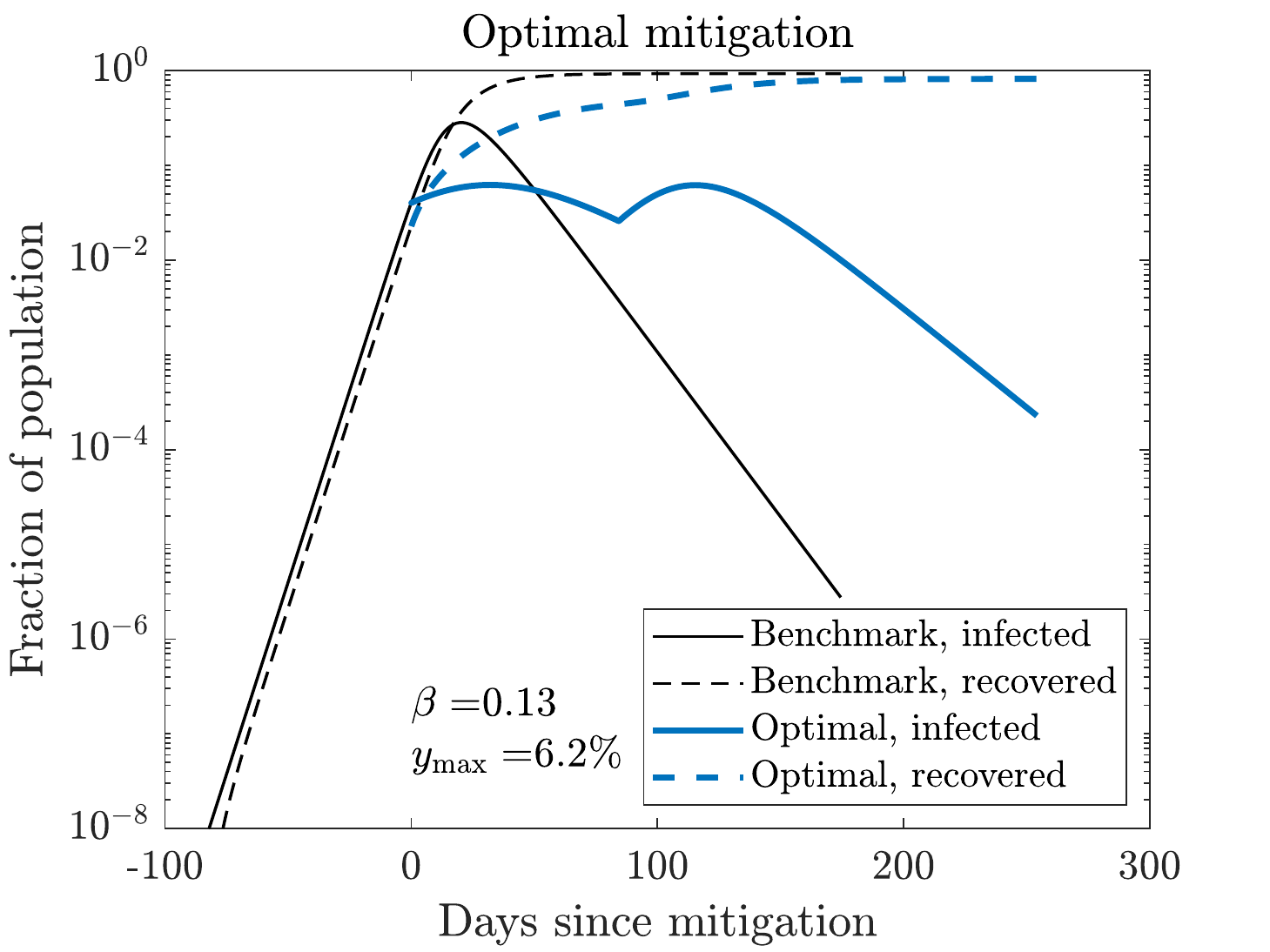}
\end{subfigure}
\caption{Dynamics of epidemic with mitigation measures.}\label{fig:mitigation}
\end{figure}

\section{Asset pricing with epidemic}

To evaluate the economic impact of the COVID-19 epidemic, in this section I solve a stylized production-based asset pricing model.\footnote{\cite{EichenbaumRebeloTrabandtEpidemics} build a quantitative macroeconomic model where economic activity (consumption and work) affects the transmission rate during an epidemic and discuss the optimal containment policy. On the empirical side, \cite{KarlssonNilssonPichler2014} find that the 1918 Spanish flu had negative effects on poverty and capital income but no effect on earnings.}

\subsection{Model}
The economy consists of two agent types, capitalists and workers, who respectively own the capital stock and labor. The capital stock at time $t$ is denoted by $K_t$. The capital growth rate is exogenous, lognormal, and i.i.d.\ over time:
$$\log (K_{t+1}/K_t) \sim N(\mu,\sigma^2).$$
Capitalists hire labor at competitive markets and produce a perishable good using a Cobb-Douglas production technology $Y=K^\alpha L^{1-\alpha}$, where $\alpha\in (0,1)$ is the capital share. The labor supply is exogenous, deterministic, and normalized to 1 during normal times. During an epidemic, workers are either susceptible, infected, or recovered, and only non-infected agents can supply labor. For simplicity, I assume that workers are hand-to-mouth and consume the entire wage. The financial market is complete, and capitalists maximize the constant relative risk aversion (CRRA) utility
$$\E_t\sum_{s=0}^\infty \beta^s\frac{C_{t+s}^{1-\gamma}}{1-\gamma},$$
where $\beta>0$ is the discount factor and $\gamma>0$ is the relative risk aversion coefficient. A stock is a claim to the representative firm's profit $K^\alpha L^{1-\alpha}-wL$, where $w$ is the wage.

Given the sequence of labor supply $\set{L_t}_{t=0}^\infty$, we can solve for the equilibrium stock price semi-analytically as follows. The first-order condition for profit maximization implies $w=(1-\alpha)(K/L)^\alpha$. Hence the firm's profit, which by market clearing must equal consumption of capitalists, is
\begin{equation}
C=K^\alpha L^{1-\alpha}-wL=\alpha K^\alpha L^{1-\alpha}.\label{eq:cons}
\end{equation}
Because the marginal buyer of the stock is a capitalist, the stochastic discount factor of the economy is given by $M_{t+1}=\beta(C_{t+1}/C_t)^{-\gamma}$. Letting $P_t$ be the stock price, the no-arbitrage condition implies
\begin{equation}
P_t=\E_t\left[\beta \left(\frac{C_{t+1}}{C_t}\right)^{-\gamma}(P_{t+1}+C_{t+1})\right].\label{eq:noarbitrage}
\end{equation}
Dividing both sides of \eqref{eq:noarbitrage} by $C_t$, letting $V_t=P_t/C_t$ be the price-dividend ratio, and using \eqref{eq:cons}, we obtain
\begin{align*}
V_t&=\E_t\left[\beta \left(\frac{C_{t+1}}{C_t}\right)^{1-\gamma}(V_{t+1}+1)\right]\\
&=\E_t\left[\beta\left((K_{t+1}/K_t)^\alpha (L_{t+1}/L_t)^{1-\alpha}\right)^{1-\gamma}(V_{t+1}+1)\right].
\end{align*}
Because capital growth is i.i.d.\ normal and labor supply is deterministic, we can rewrite the price-dividend ratio as
\begin{equation}
V_t=\kappa(L_{t+1}/L_t)^{(1-\alpha)(1-\gamma)}(V_{t+1}+1),\label{eq:PDratio}
\end{equation}
where $\kappa=\beta \e^{\alpha(1-\gamma)\mu+[\alpha(1-\gamma)]^2\sigma^2/2}$. In normal times, we have $L_t\equiv 1$ and $V_t\equiv \frac{\kappa}{1-\kappa}$, where we need $\kappa<1$ for convergence. During an epidemic, it is straightforward to compute the price-dividend ratio by iterating \eqref{eq:PDratio} using the boundary condition $V_\infty=\frac{\kappa}{1-\kappa}$.

\subsection{Calibration}

I calibrate the model at daily frequency. I set the capital share to $\alpha=0.38$ and the relative risk aversion to $\gamma=3$, which are standard values. I assume a 4\% annual discount rate, so $\beta=\exp(-0.04/N_d)$, where $N_d=365.25$ is the number of days in a year. To calibrate capital growth and volatility, note that in normal times we have $L=1$ and hence $Y=K^\alpha$. Taking the log difference, we obtain $\log (Y_{t+1}/Y_t)=\alpha \log (K_{t+1}/K_t)$. Therefore according to the model, capital growth rate and volatility are $1/\alpha$ times those of output. I calibrate these parameters from the US quarterly real GDP per capita in 1947Q1--2019Q4 and obtain $\mu=0.0511$ and $\sigma=0.0487$ at the annual frequency.\footnote{At daily frequency, we need to divide these numbers by $N_d$ and $\sqrt{N_d}$, respectively.} For the transmission rate, using the point estimates in Section \ref{sec:estim}, I consider $\beta_0=\betaMed$. The recovery rate is $\gamma_0=0.1$. The initial condition is $(y_0,z_0)=(10^{-8},0)$.

Figure \ref{fig:asset_price} shows the stock price relative to potential output $P_t/Y_t^*$, where $Y_t^*=K_t^\alpha$ is the full employment output. The left and right panels are under the benchmark case and optimal policy, respectively. In the benchmark model, the stock price decreases sharply during the epidemic by about 50\%. However, the stock market crash is short-lived and prices recover quickly after the epidemic. This observation is in sharp contrast to the prediction from rare disasters models \citep{rietz1988,Barro2006QJE}, where shocks are permanent. Under the optimal policy, because the infection rate $y$ has two peaks, the stock price shows a W-shaped pattern. However, the decline is much more moderate at around 10\%.

\begin{figure}
\centering
\begin{subfigure}{0.48\linewidth}
\includegraphics[width=\linewidth]{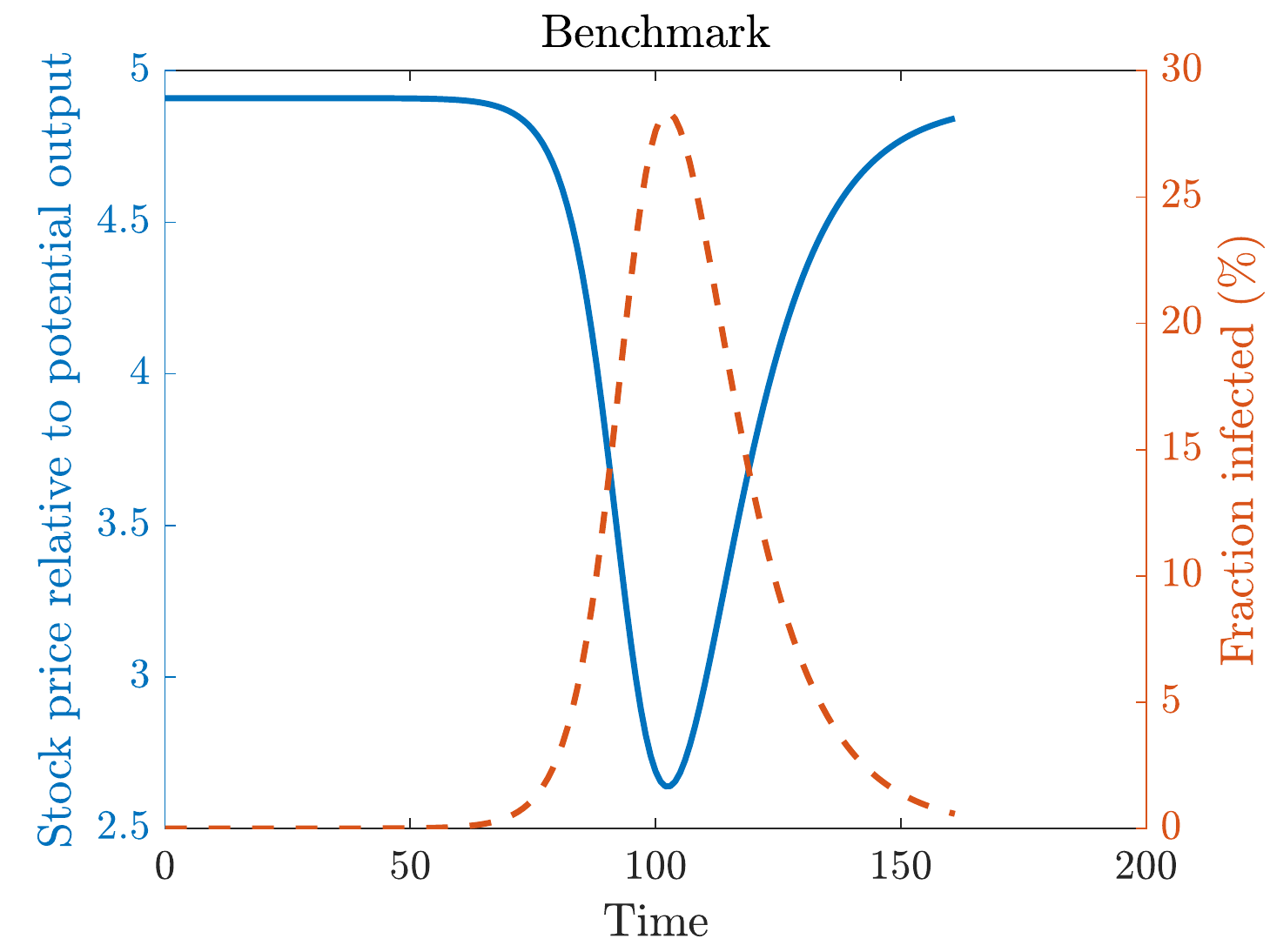}
\end{subfigure}
\begin{subfigure}{0.48\linewidth}
\includegraphics[width=\linewidth]{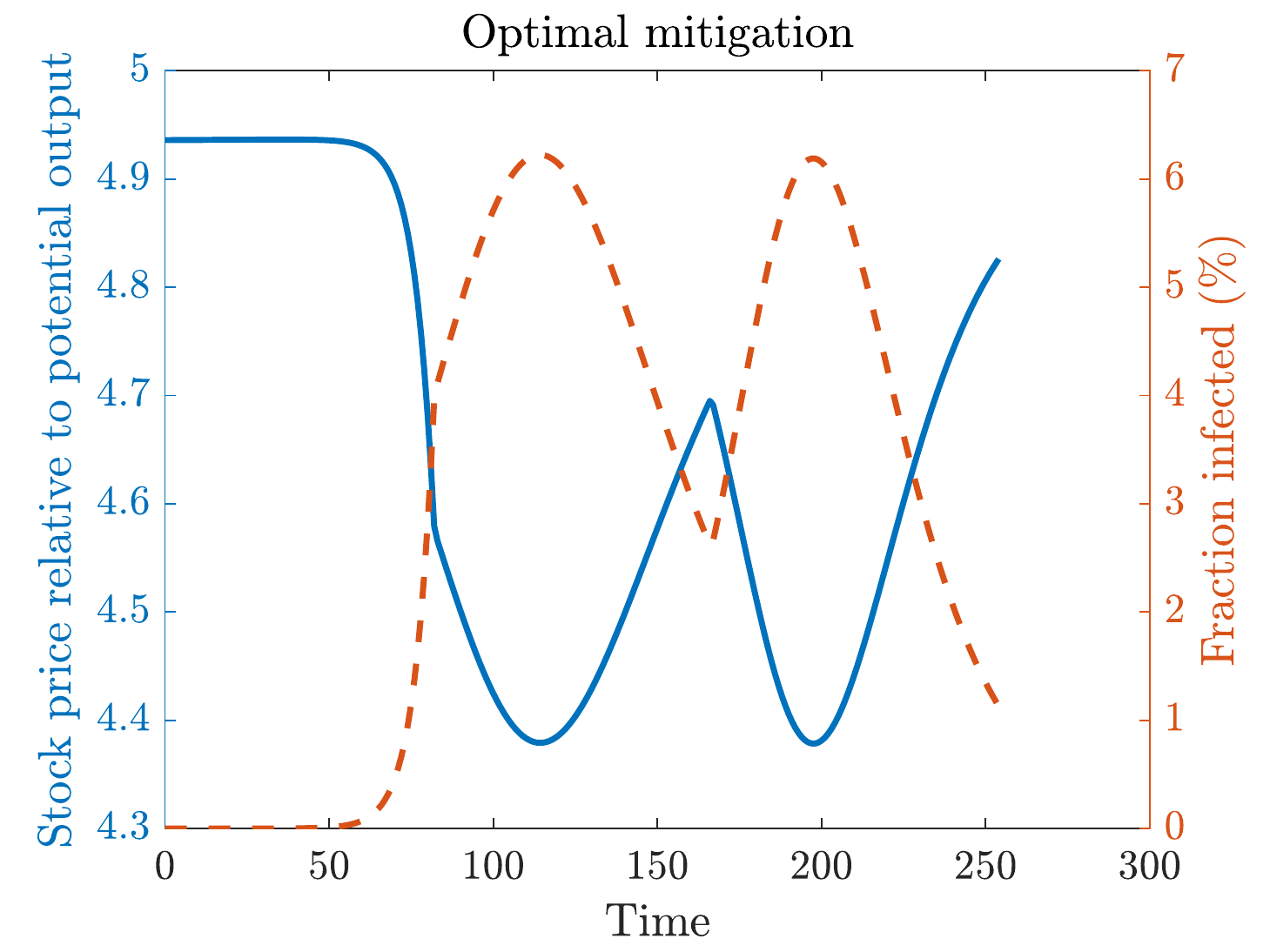}
\end{subfigure}
\caption{Asset prices during epidemic.}\label{fig:asset_price}
\end{figure}

\section{Conclusion}

Because the situation with COVID-19 is rapidly evolving, any analysis based on current data will quickly become out of date. However, any analysis based on available data is better than no analysis. With these caveats in mind, I draw the following conclusions from the present analysis.

The COVID-19 epidemic is spreading except in China, Japan, and Korea. In many countries the transmission rate at present (\today) is very high at around $\beta=0.3$. This number implies that it takes only $1/\beta\approx 3$ days for a patient to infect another individual. Since it takes around 10 days to recover from the illness, the number of patients will grow exponentially and may overwhelm the healthcare system if no actions are taken. If the current trend continues, the epidemic will peak in early May 2020 in Europe and North America, at which point around 30 percent of the population will be infected. Because the recovery rate $\gamma$ is an uncontrollable biological parameter, the only way to control the epidemic is to reduce the transmission rate $\beta$, perhaps by restricting travel or social distancing. However, temporary measures only slows the onset of the epidemic but has no effect on the peak because the epidemic persists until the population acquires herd immunity. The optimal policy that minimizes the peak is to wait to introduce mitigation measures until a sufficient fraction of the population is infected, which can reduce the peak to 6.2\%. Policy makers in affected countries may also want to look at measures taken in China, Japan, and Korea, which are the countries relatively successful at controlling the spread so far.

Using the estimated transmission rates, I have solved a stylized production-based asset pricing model. The model predicts that the stock price decreases by 50\% during the epidemic, but recovers quickly afterwards because the epidemic is a short-lived labor supply shock. Under the optimal policy, the stock price exhibits a W-shaped pattern and remains about 10\% undervalued than the steady state level for half a year.

\newpage

%\bibliographystyle{plainnat}
%\bibliography{reference}

\appendix

\section{Solving the SIR model numerically}\label{sec:solve}

In principle, solving the SIR model numerically is straightforward using the following algorithm.
\begin{enumerate}
\item Given the parameters $(\beta,\gamma)$ and initial condition $(x_0,y_0)$, solve for $v^*$ as the unique solution to \eqref{eq:vstar}.
\item Take a grid $1=v_0>v_1>v_2>\dots>v_N>v^*$. For each $n=1,\dots,N$, compute the integral
\begin{equation}
I_n=\int_{v_n}^{v_{n-1}}\frac{\diff \xi}{\xi(\beta x_0(1-\xi)+\beta y_0+\gamma\log \xi)}\label{eq:In}
\end{equation}
numerically.
\item Define $t_0=0$ and $t_n=\sum_{k=1}^nI_n$ for $n\ge 1$. Compute $(x_n,y_n,z_n)$ using \eqref{eq:HLM} evaluated at $v=v_n$. Then $\set{t_n,(x_n,y_n,z_n)}_{n=0}^N$ gives the numerical solution to the SIR model.
\end{enumerate}

Although the above algorithm is conceptually straightforward, there are two potential numerical issues. First, the integrand
\begin{equation}
g(\xi)\coloneqq \frac{1}{\xi(\beta x_0(1-\xi)+\beta y_0+\gamma \log\xi)}\label{eq:integrand}
\end{equation}
in \eqref{eq:v} is not well-behaved near $\xi=1$. In fact, setting $\xi=1$ we obtain $g(1)=1/\beta y_0$, which is typically a very large number since $y_0$ (the fraction of infected at $t=0$) is typically small, say of the order $10^{-6}$. This makes the numerical integral $I_n$ in \eqref{eq:In} inaccurate for small $n$. Second, for applications we would like the dates $\set{t_n}_{n=0}^N$ to be well-behaved (say approximately evenly spaced), which requires an appropriate choice of the grid $\set{v_n}_{n=0}^N$.

To deal with the first issue, let us express $g$ as $g=h_1+h_2$, where $h_1$ has a closed-form primitive function and $h_2$ is well-behaved near $\xi=1$. Since $\log \xi\approx \xi-1$ near $\xi=1$, a natural candidate is
\begin{align*}
h_1(\xi)&\coloneqq \frac{1}{\xi(\beta x_0(1-\xi)+\beta y_0+\gamma(\xi-1))}\\
&=
\begin{cases}
\frac{1}{\beta(x_0+y_0)-\gamma}\left(\frac{1}{\xi}+\frac{\beta x_0-\gamma}{(\beta x_0-\gamma)(1-\xi)+\beta y_0}\right), & (\beta(x_0+y_0)\neq \gamma)\\
\frac{1}{\beta y_0 \xi^2}. & (\beta(x_0+y_0)=\gamma)
\end{cases}
\end{align*}
Then by simple algebra, \eqref{eq:v} becomes
\begin{multline}
t=\int_v^1h_2(\xi)\diff \xi\\
+\begin{cases}
\frac{1}{\beta(x_0+y_0)-\gamma}\log \frac{(\beta x_0-\gamma)(1-v)+\beta y_0}{\beta y_0v}, & (\beta(x_0+y_0)\neq \gamma)\\
\frac{1}{\beta y_0}\left(\frac{1}{v}-1\right), & (\beta(x_0+y_0)=\gamma)
\end{cases}\label{eq:v2}
\end{multline}
where
\begin{equation}
h_2(\xi)\coloneqq \frac{1}{\xi(\beta x_0(1-\xi)+\beta y_0+\gamma \log\xi)}-\frac{1}{\xi((\beta x_0-\gamma)(1-\xi)+\beta y_0)}.\label{eq:h2}
\end{equation}
Because $h_2(\xi)$ is approximately 0 to the first order around $\xi=1$, we can calculate the numerical integrals in \eqref{eq:In} accurately.

To deal with the second issue, consider the SIR model \eqref{eq:xyz} with $\gamma=0$. Then \eqref{eq:xyz.x} becomes $\dot{x}=-\beta x(1-x)$, and by separation of variables we obtain the analytical solution
$$x(t)=\frac{x_0}{x_0+(1-x_0)\e^{\beta t}}.$$
Using \eqref{eq:HLM.x}, for the case $\gamma=0$, time $t$ and parameter $v$ are related as
$$v=\frac{1}{x_0+(1-x_0)\e^{\beta t}}.$$
Define $t^*$ by
$$v^*=\frac{1}{x_0+(1-x_0)\e^{\beta t^*}}\iff t^*=\frac{1}{\beta}\log \frac{1/v^*-x_0}{1-x_0}.$$
Finally, define
$$v_n=\frac{1}{x_0+(1-x_0)\e^{\beta t^*n/N}}.$$
Then $t_n$ implied by \eqref{eq:v} is evenly spaced when $\gamma=0$, and we can expect that the grid $\set{v_n}_{n=0}^N$ gives reasonable values of $\set{t_n}_{n=0}^N$ even when $\gamma>0$.

For numerical implementation, I set $N=1000$ and use the 11-point Gauss-Legendre quadrature and \eqref{eq:v2} to numerically compute the integral in \eqref{eq:In}. Figure \ref{fig:SIR_example} shows the dynamics of SIR model when $(\beta,\gamma)=(0.2,0.1)$, $y_0=10^{-6},10^{-5},10^{-4}$, and $z_0=0$. For this example, $1-v^*=80.0\%$ of the population is eventually infected, and $y_{\max}=15.4\%$ of the population is infected at the peak of the epidemic. The initial condition ($y_0$) affects the timing of the epidemic but not its dynamics.

\begin{figure}[!htb]
\centering
\includegraphics[width=0.7\linewidth]{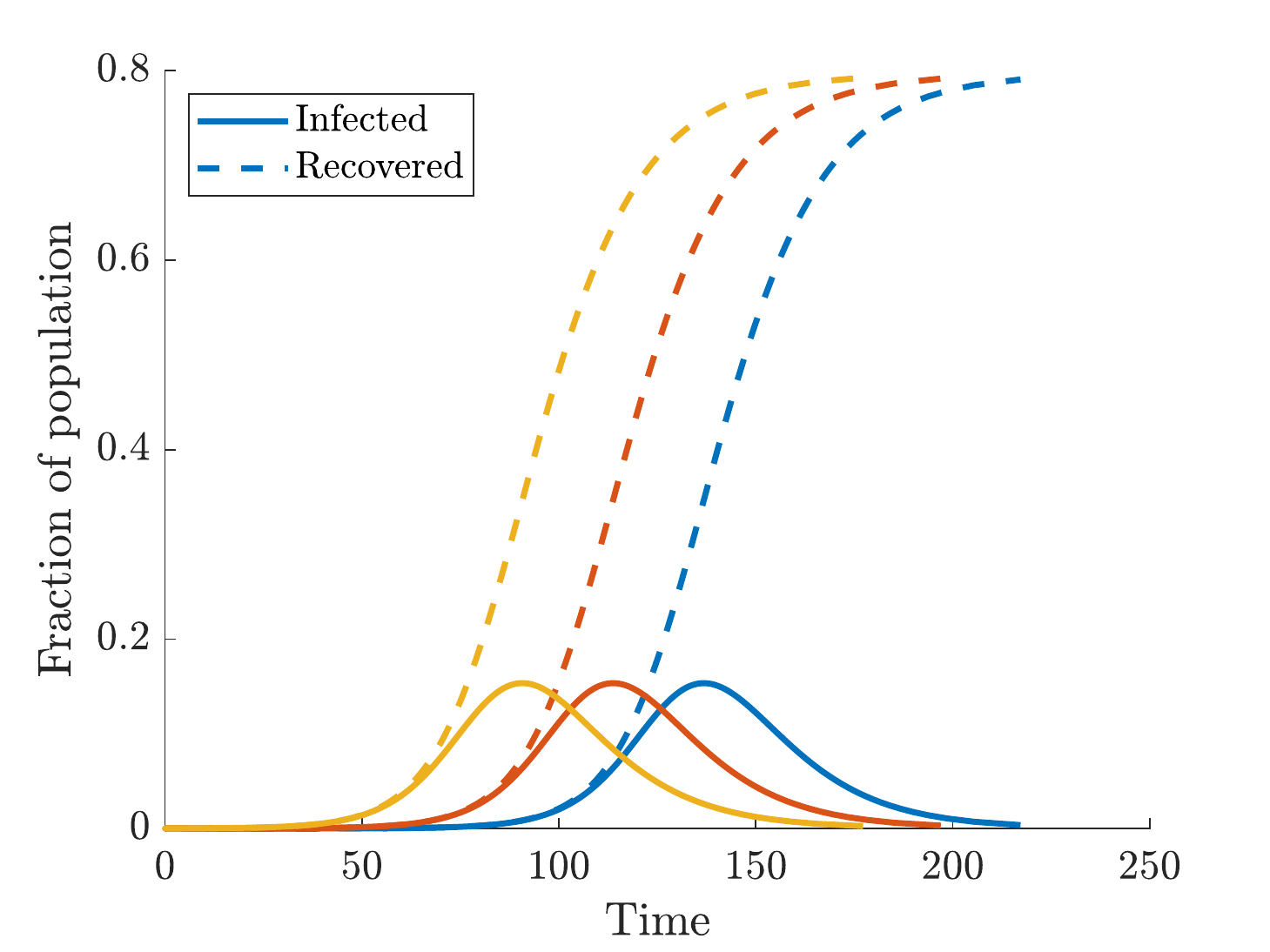}
\caption{Dynamics of SIR model when $(\beta,\gamma)=(0.2,0.1)$, $y_0=10^{-6},10^{-5},10^{-4}$, and $z_0=0$. Smaller $y_0$ corresponds to later onset of epidemic.}\label{fig:SIR_example}
\end{figure}

\end{document}